\RequirePackage{fix-cm}
\documentclass[a4paper,10pt]{article}

\usepackage{fullpage}

\usepackage{graphicx}

\usepackage{amsthm}
\usepackage{amssymb,amsfonts,amsmath}
\usepackage {color}

\usepackage{hyperref}

\newcommand{\seq}[1]{\ensuremath{\left \langle #1 \right \rangle}}
\newcommand{\A}{\ensuremath{\mathcal{A}}}
\newcommand{\euclid}{\ensuremath{\mathbb{E}}}
\newcommand{\abs}[1]{\ensuremath{\left | #1 \right |}}
\newcommand{\floor}[1]{\ensuremath{\left \lfloor #1 \right \rfloor}}

\newcommand{\fig}[1]{\figurename~\ref{#1}}

\graphicspath{{figures/}}

\usepackage{wrapfig}

\makeatletter
\g@addto@macro\bfseries{\boldmath}
\makeatother

\title{Ham-Sandwich Cuts for Abstract Order Types%
\thanks{This work has been supported by the ESF EUROCORES programme EuroGIGA - ComPoSe.
A.P.\ is supported by Austrian Science Fund (FWF): I 648-N18.
Part of this work has been done while he was recipient of a DOC-fellowship of the Austrian Academy of Sciences at the Institute for Software Technology, Graz University of Technology, Austria.
A preliminary version of this paper appeared in the proceedings of ISAAC 2014~\cite{fp-hscao-14}.
Part of this work was presented in the PhD thesis~\cite{complexity_order_types} of the second author.}}

\author{Stefan~Felsner\thanks{Institut f\"ur Mathematik, Technische Universit\"at Berlin, Germany.
\texttt{felsner@math.tu-berlin.de}.}
\and Alexander~Pilz\thanks{Institute for Software Technology,
Graz University of Technology,
Austria.
\texttt{apilz@ist.tugraz.at}.}}

\date{\today}

\newtheorem{theorem}{Theorem}
\newtheorem{observation}{Observation}

\newtheorem{corollary}{Corollary}
\newtheorem{definition}{Definition}
\newtheorem{problem}{Problem}
\newtheorem{lemma}{Lemma}

\begin{document}

\maketitle

\begin{abstract}
\sloppypar{
The linear-time ham-sandwich cut algorithm of Lo, Matou\v{s}ek, and Steiger for bi-chromatic finite point sets in the plane works by appropriately selecting crossings of the lines in the dual line arrangement with a set of well-chosen vertical lines.
We consider the setting where we are not given the coordinates of the point set, but only the orientation of each point triple (the order type) and give a deterministic linear-time algorithm for the mentioned sub-algorithm.
This yields a linear-time ham-sandwich cut algorithm even in our restricted setting.
We also show that our methods are applicable to abstract order types.
}
\end{abstract}

\section{Introduction}
Goodman and Pollack investigated ways of partitioning the infinite number of point sets in the plane into a finite number of equivalence classes.
To this end they introduced \emph{circular sequences}~\cite{nondegenerate} and \emph{order types}~\cite{GoodmanPollack83}.
Two point sets $S_1$ and $S_2$ have the same \emph{order type} iff there exists a bijection between the sets s.t.\ a triple in~$S_1$ has the same orientation (clockwise or counterclockwise) as the corresponding triple in~$S_2$
(in this paper we only consider point sets in general position, i.e., no three points are collinear).

The order type determines many important properties of a point set, e.g., its convex hull and which spanned segments cross.
Determining the orientation of a point triple (called a \emph{sidedness query}) can be done in a computationally robust way~\cite{sign_of_determinants}.
Therefore, algorithms that base their decisions solely on sidedness queries allow robust implementations~\cite{boissonnat_snoeyink}.
Furthermore, this restriction is helpful for mechanically proving correctness of algorithms~\cite{formal_isabelle,formal_coq}.
Another possible advantage of restricting algorithms to sidedness queries allows, for some problems, to cope with degeneracies in point sets.
Edelsbrunner and M\"ucke~\cite{simplicity} describe a general framework for the so-called ``simulation of simplicity'';
in particular, they provide an efficient way to conceptually replace a point set with collinear point triples by one in general position s.t.\ all other orientations of point triples are preserved.

The duality between point sets and their dual line arrangements is a well-established tool in discrete and computational geometry.
Line arrangements can be generalized to pseudo-line arrangements, and many combinatorial and algorithmic questions that can be asked for line arrangements are also interesting for pseudo-line arrangements.
The order type of a point set is encoded in the structure of the dual line arrangement of a point set, in particular by the lines (and even by the number of lines~\cite{GoodmanPollack83}) above and below a crossing in the dual arrangement.
See~\cite{new_trends_gp} for an in-depth treatment of that topic.
The orientation of a triple of pseudo-lines can be obtained from the ordering of crossings just as for lines.
The triple-orientations fulfill certain axioms, and concepts like the convex hull can be defined for sets with appropriate triple-orientations~\cite{knuth} even though they may not be realizable by a point set.
Such a generalization of order types is known as \emph{abstract order type}.
Besides their combinatorial properties, algorithmic aspects of abstract order types have been studied.
Knuth devotes a monograph~\cite{knuth} to this generalization and its variants, in particular w.r.t.\ convex hulls.
Motivated by Knuth's open problems, Aichholzer, Miltzow, and Pilz~\cite{extreme_journal} show how to obtain, for a given pair $(a,b)$ of an abstract order type, the edges of the convex hull that are intersected by the supporting line of $ab$ in linear time, using only sidedness queries.
There appears to be no known reasonable algorithmic problem that can be formulated for both order types and abstract order types such that there is an algorithm for order types that is asymptotically faster than any possible algorithm for abstract order types (see also the discussion in~\cite[p.~29]{erickson_thesis}).
In this paper, we show that the ham-sandwich cut is another problem that does not provide such an example.
Apart from being of theoretical interest, abstract order types that are not realizable by point sets occur naturally when the point set is surrounded by a simple polygon and point triples are oriented w.r.t.\ the geodesics between them~\cite{geodesic_ot}.

A considerable fraction of problems in computational geometry deals with partitioning finite sets of points by hyperplanes while imposing constraints on both the subsets of the partition as well as on the hyperplanes.
In the plane, examples of this class of problems are finding, e.g., a ham-sandwich cut of a bi-chromatic point set~\cite{lo_matousek_steiger}, a four-way partitioning by orthogonal lines, and a six-way partitioning by three concurrent lines~\cite{roy_steiger}, as well as finding three concurrent halving lines that pairwise span an angle of $60^\circ$~\cite{wedges}.

Given a pair $(a,b)$ of points of a bi-chromatic point set $S$ of~$n$ points that are either red or blue, the supporting line of~$a$ and~$b$ is a \emph{ham-sandwich cut} if not more than half of the red and half of the blue points are on either side of $ab$.
This can be verified by using only sidedness queries (implying a brute-force algorithm running in $\Theta(n^3)$ time).
Megiddo~\cite{megiddo} presented a linear-time algorithm for the case in which the points of one color are separable from the points of the other color by a line.
Edelsbrunner and Waupotitsch~\cite{edelsbrunner_waupotitsch} gave an $O(n \log (\min\{n_\mathrm{r}, n_\mathrm{b}\}))$ time algorithm for the general case, with $n_\mathrm{r}$ red and $n_\mathrm{b}$ blue points.
Eventually, a linear-time algorithm was provided by Lo, Matou\v{s}ek, and Steiger~\cite{lo_matousek_steiger} for the general setting (abbrev.~\emph{LMS algorithm}).
Their approach generalizes to arbitrary dimensions.
Bose et al.~\cite{hamsan} generalize ham-sandwich cuts to points inside a simple polygon, obtaining a randomized $O((n+m) \log r)$ time algorithm, where $m$ is the number of vertices of the polygon, of which $r$ are reflex.
Ham-sandwich cuts belong to a class of problems in computational geometry that deal with partitioning finite sets of points by hyperplanes while imposing constraints on both the subsets of the partition as well as on the hyperplanes; see, e.g.,~\cite{roy_steiger} for algorithms for related problems.

The LMS algorithm works on the dual line arrangement of the point set and has to solve the following sub-problem.

\begin{problem}\label{problem_main_geometric}
Given a line arrangement~$\A$ in the plane and two lines~$p$ and~$q$ of that arrangement, let~$v$ be the vertical line passing through the crossing of~$p$ and~$q$.
For a subset $B$ of the lines in~$\A$ and an integer $k \leq |B|$, find a line $m \in B$ such that the $y$-coordinate of the point $v \cap m$ is of rank~$k$ in the sequence of $y$-coordinates of the finite point set $v \cap \bigcup_{b \in B} b$.
\end{problem}

\begin{wrapfigure}{r}{0.5\textwidth}
\centering
\includegraphics[width=0.48\textwidth]{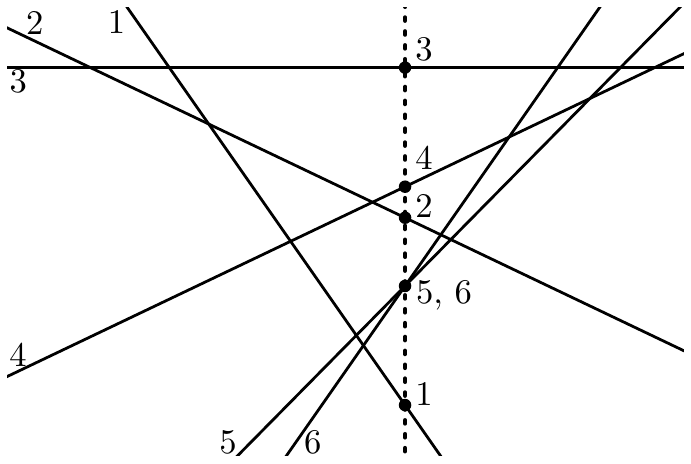}
\caption{Crossings along a vertical line.}
\label{fig_vertical_selection}
\end{wrapfigure}
This problem can be solved in linear time by directly applying the linear-time selection algorithm~\cite{blum} to the $y$-coordinates of the intersections of all lines in $B$ with~$v$.
Clearly, the order of the intersections of lines with a vertical line at a crossing is not a property of the order type represented by the arrangement (e.g., in \fig{fig_vertical_selection}, one could change the order of the lines~2 and 4 above the crossing between line 5 and line 6).
The order type only determines the set of lines above and below a crossing.

A reformulation of Problem~\ref{problem_main_geometric} for abstract order types faces two problems.
First, the vertical direction is not determined by the order type (this is a property of the circular sequence of the point set); even though we can represent the abstract order type by a pseudo-line arrangement in the Euclidean plane (where there is a vertical direction), there is, in general, an exponential number of different ways to draw a pseudo-line arrangement representing the abstract order type (w.r.t.\ the $x$-order of crossings), yielding too many different orders of the pseudo-lines along the vertical line through a crossing.
Second, even when such a vertical line is given, directly applying the linear-time selection algorithm requires that a query comparing the order of two pseudo-lines on the vertical line can be answered in constant time.

In this paper, we show how to overcome these two problems.
We define a ``vertical'' pseudo-line through each crossing in a pseudo-line arrangement and show how to select the pseudo-line of a given rank in the order defined by such a ``vertical'' pseudo-line.
We give the precise definition in Section~\ref{sec_levels_at_crossing} where we also examine properties of the construction.
The result is presented in terms of a (dual) pseudo-line arrangement in the Euclidean plane~$\euclid^2$.
However, in our model we are not given an explicit representation but are only allowed sidedness queries.
In Section~\ref{sec_levels_algorithm}, we first explain how the queries about a pseudo-line arrangement can be mapped to sidedness queries, and then give a linear-time algorithm for selecting a pseudo-line with a given rank.
Our result allows for replacing vertical lines in the LMS algorithm, showing that it also works for abstract order types.
An analysis of the LMS algorithm under this aspect is given in Section~\ref{sec_ham_sandwich_revisited}.
We obtain

\begin{theorem}\label{thm_abstract_ham_sandwich}
A ham-sandwich cut of an abstract order type can be found in linear time using only sidedness queries.
\end{theorem}

The observation that the LMS algorithm in principle also works for pseudo-line arrangements has been used by Bose et al.~\cite{hamsan} for their randomized linear-time algorithm for geodesic ham-sandwich cuts.
However, their pseudo-lines are given by (weakly) $x$-monotone polygonal paths with a constant number of edges.
Hence, the intersection of such a path with a vertical line can be computed in constant time, like in the straight-line setting.
Their randomized algorithm runs in $O((n+m) \log r)$ time, where $n$ is the number of red and blue points, $m$ is the number of vertices of the polygon, of which $r$ are reflex.
Geodesic order types are a subset of abstract order types~\cite{geodesic_ot}.
When applying a result from~\cite{simplification} to get, after $O(m)$ preprocessing time, the orientation of each triple of points in a simple polygon in $O(\log r)$ time in combination with the ham-sandwich cut algorithm for abstract order types, we obtain a deterministic $O(n \log r + m)$ time algorithm for geodesic ham-sandwich cuts ``for free''.
(Note that this does not contradict optimal worst-case behavior shown by Bose et al.~\cite{hamsan} for their algorithm, as their analysis is parameterized by $(n+m)$ and~$r$.)
We emphasize that a detailed analysis of their approach may give a more fine-grained runtime analysis, and may allow for directly applying common derandomization techniques.
Nevertheless, our technique results in a complete separation of the part that is specific to the geodesic setting, implementing a general subroutine, and the ham-sandwich cut algorithm.

\newcommand{\GA}[1]{\ensuremath{\gamma_{#1}}}
\newcommand{\lv}{\ensuremath{\mathrm{lv}}}
\newcommand{\rk}{\ensuremath{\mathrm{rk}}}

A \emph{pseudo-line} is an $x$-monotone plane curve in~$\euclid^2$.
A \emph{pseudo-line arrangement} is the cell complex defined by the dissection of $\euclid^2$ by a set of pseudo-lines such that each pair of pseudo-lines intersects in exactly one point, at which they cross.
An arrangement is \emph{simple} if no three pseudo-lines intersect in the same point.
Throughout this paper, let~$\A$ be a simple arrangement of~$n$ pseudo-lines.
The two vertically unbounded cells in~$\A$ are called the \emph{north face} and the \emph{south face}.
The \emph{$k$-level} of~$\A$ is the set of all points that lie on a pseudo-line of~$\A$ and have exactly $k-1$ pseudo-lines strictly above them.
The level of a crossing $pq$ is denoted by $\lv(pq)$ (i.e., $pq$ is separated from the north face by $\lv(pq)-1$ pseudo-lines).
The \emph{upper envelope} of an arrangement is its 1-level, i.e., the union of the segments of pseudo-lines that are incident to the north face.

\section{Pseudo-Verticals}\label{sec_levels_at_crossing}
It will be convenient to consider all pseudo-lines being directed towards positive $x$-direction.
Let~$p$ and~$q$ be two pseudo-lines in~$\A$ and let~$p$ start above~$q$.
We denote the latter by $p \prec q$.
Our first aim is to define a pseudo-vertical through a crossing, i.e,  an object that can be used like a vertical line through a crossing in our abstract setting.

For a crossing $pq$ with $p \prec q$ let $\GA{pq}$ be a curve described by the following local properties.
Initially, $\GA{pq}$ passes through the crossing $pq$ and enters the cell~$C$ directly above~$pq$; see \fig{fig_local_vertical}~(a).
To define $\GA{pq}$ it is convenient to think of it as consisting of two parts.
The \emph{northbound ray} is the part starting at $pq$ leading to the north face while the \emph{southbound ray} connects $pq$ to the south face.
Starting from $pq$ the northbound ray follows~$p$ against its direction moved slightly into the interior of cell $C$.
In general the northbound ray of $\GA{pq}$ will be slightly above some line $a_i$ moving against the direction of~$a_i$.
When $a_i$ is crossed by a pseudo-line~$a_j$ we have two cases.
If $a_i$ is crossed from below, $\GA{pq}$ also crosses~$a_j$, and continues following $a_i$; see \fig{fig_local_vertical}~(b).
If~$a_i$ is crossed by~$a_j$ from above, $\GA{pq}$ leaves $a_i$ and continues following~$a_j$ against its direction; see \fig{fig_local_vertical}~(c).
This is continued until $\GA{pq}$ follows some line~$a_i$ and all crossings of~$A$ are to the right of the current position along~$\GA{pq}$.
At that point, $\GA{pq}$ continues vertically in positive $y$-direction to infinity (i.e., it crosses all lines $a$ with $a\prec a_i$ in decreasing $\prec$-order); see \fig{fig_local_vertical_last}~(a).
The southbound ray of $\GA{pq}$ is defined in a similar manner.
It follows some pseudo-lines in their direction but slightly below.
It starts with~$p$ and when following $a_i$ it changes to $a_j$ at a crossing only when $a_j$ is crossing from above (see \fig{fig_local_vertical}~(d) and \fig{fig_local_vertical}~(e)).
The final part may again consist of some crossings with lines $a$ with $a\prec a_i$ in decreasing $\prec$-order.
\fig{fig_gamma_whole_example_straight} gives an example of a pseudo-vertical in the arrangement of \fig{fig_vertical_selection}, and \fig{fig_gamma_whole_example} shows a pseudo-vertical in a wiring diagram.

\begin{figure}[ht]
\centering
\includegraphics[scale=.91]{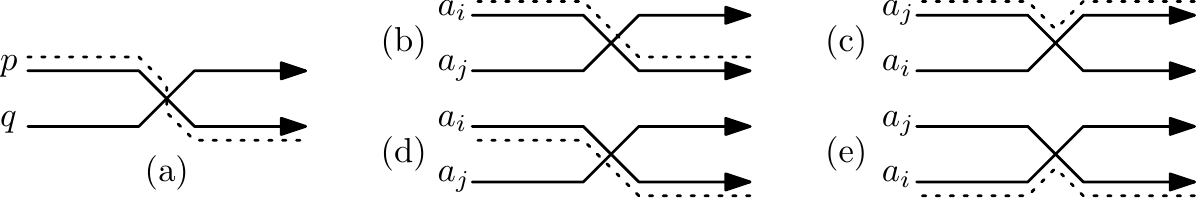}
\caption[Local definition of a pseudo-vertical~$\GA{pq}$.]{
Local definition of a pseudo-vertical~$\GA{pq}$.}
\label{fig_local_vertical}
\end{figure}

\begin{figure}
\centering
\includegraphics{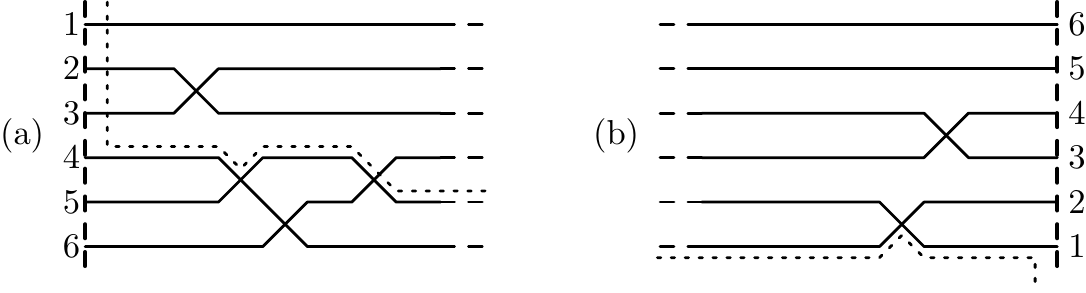}
\caption[The first and the last pseudo-lines of an arrangement defining~$\GA{pq}$.]{The first~(a) and the last~(b) pseudo-lines of an arrangement defining~$\GA{pq}$~(dotted).%
}
\label{fig_local_vertical_last}
\end{figure}

\begin{figure}[ht]
\centering
\includegraphics{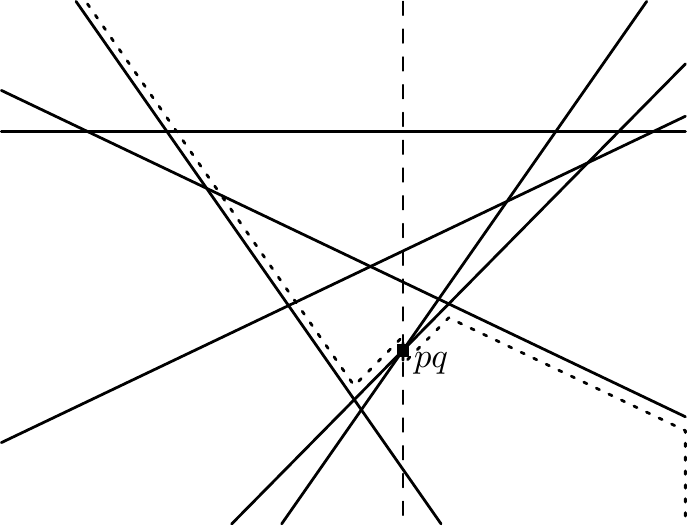}
\caption{A pseudo-vertical~$\GA{pq}$ (dotted) in an arrangement of straight lines.}
\label{fig_gamma_whole_example_straight}
\end{figure}

We call $\GA{pq}$ a \emph{pseudo-vertical} and, in the following, identify several properties of such a curve.
Note that, while we used the (rather informal) notion of ``following'' a pseudo-line, $\GA{pq}$ is actually defined by the cells it traverses (i.e., two paths in the dual graph of the cell complex starting at the cells above and below~$pq$).
As~$\GA{pq}$ always follows a pseudo-line of $\A$ or continues in a vertical direction, we note that $\GA{pq}$ is $x$-monotone.

\begin{figure}
\centering
\includegraphics[width=.9\textwidth]{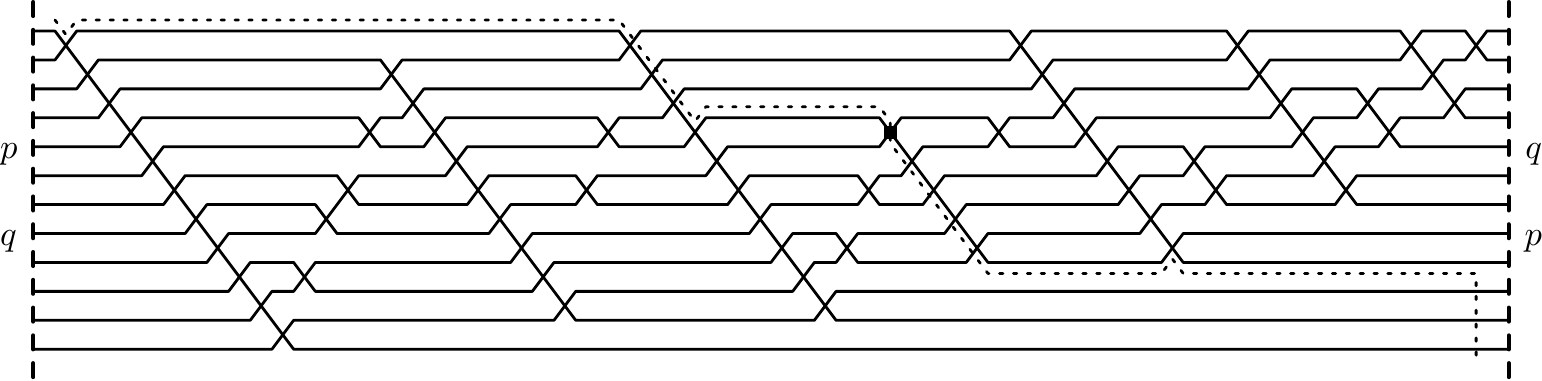}
\caption{A pseudo-vertical~$\GA{pq}$ in a pseudo-line arrangement.
The $x$-order or the crossings represents the order induced by the pseudo-verticals.}
\label{fig_gamma_whole_example}
\end{figure}

\subsection{Properties of a Pseudo-Vertical}
As~$\GA{pq}$ always follows a pseudo-line of $\A$ or continues in a vertical direction, we can observe the following.
\begin{observation}
For any crossing~$pq$ in a pseudo-line arrangement~$\A$, the curve $\GA{pq}$ is $x$-monotone.
\end{observation}

The following observation can easily be made by visualizing the arrangement as a wiring diagram.

\begin{observation}\label{obs_level_increasing}
The number of pseudo-lines above a point moving along~$\GA{pq}$ in positive $x$-direction is a monotone function, it increases at every crossing of $\GA{pq}$ with a pseudo-line of~$\A$.
\end{observation}
\begin{lemma}\label{lem_vertical_pseudo_line}
For any crossing~$pq$ in a pseudo-line arrangement~$\A$, the curve $\GA{pq}$ is a pseudo-line such that $\A$ can be extended by~$\GA{pq}$ to a new (non-simple) pseudo-line arrangement.
\end{lemma}
\begin{proof}
Let~$n$ be the number of pseudo-lines in~$\A$.
Since $\GA{pq}$ continues to vertical infinity in both positive and negative $y$-direction, it crosses every pseudo-line of~$\A$ at least once.
From Observation~\ref{obs_level_increasing}, it follows that~$\GA{pq}$ crosses at most~$n$ pseudo-lines.
As~$\GA{pq}$ is an $x$-monotone curve that crosses each pseudo-line of~$\A$ exactly once, an extension of~$\A$ is again a (non-simple) pseudo-line arrangement.
\end{proof}

We say that pseudo-line $a$ is above a crossing $pq$ if $a$ is intersected by the northbound ray of $\GA{pq}$.
If $a$ is intersected by the southbound ray of $\GA{pq}$ it is considered to be below $pq$.
(Note that this is equivalent to $a$ separating $pq$ from the north face or the south face, respectively.)
Just like a vertical line in a line arrangement, a pseudo-vertical defines a total order on the pseudo-lines of~$\A$ by the order it crosses them.
We denote the rank of a pseudo-line $m \in \A$ in this order by $\rk_{{pq}}(m)$.
The following lemma shows how we can determine the rank of an element.

Let $L(pq)$ be the set of pseudo-lines in~$\A$ such that each~$a \in L(pq)$ is below $pq$ and $a \prec p$.

\begin{lemma}\label{lem_upper_envelope}
The northbound ray of $\GA{pq}$ starting from the crossing $pq$ until reaching an unbounded cell for the first time, follows the upper envelope of the sub-arrangement defined by \mbox{$L(pq) \cup \{p\}$}.
\end{lemma}
\begin{proof}
The proof is by induction on the sequence $\seq{p=a_1, a_2, \dots }$ of pseudo-lines that we follow.
Clearly, the point $pq$ is on the upper envelope of~$L(pq) \cup \{p\}$.
Suppose we traverse $\GA{pq}$ in negative $x$-direction, following a pseudo-line $a_i \in L(pq) \cup \{p\}$.
If~$\GA{pq}$ crosses a pseudo-line~$r$ (i.e., $r$ crosses $a_i$ from below), then~$r$ cannot be below $pq$ as it would have to cross~$\GA{pq}$ again.
If a pseudo-line $a_{i+1}$ crosses $a_i$ from above, then $a_{i+1}$ cannot be above $pq$ as it would have to cross $a_i$ again.
Further, $\GA{pq}$ continues on $a_{i+1}$, keeping the invariant that no element of $L(pq)$ is above the point traversing~$\GA{pq}$.
\end{proof}

Note that every pseudo-line that passes through the upper envelope (from below) will cross~$\GA{pq}$ immediately after that crossing.

\begin{corollary}\label{cor_envelope_order}
Let~$m$ be a pseudo-line in~$\A$ that is above $pq$ and for which there exists a pseudo-line $a \in L(pq) \cup \{p\}$ such that $a \prec m$, i.e., $m$ crosses the upper envelope of $L(pq) \cup \{p\}$ by crossing some $e \in L(pq) \cup \{p\}$ with $e \prec m$.
Then the rank $\rk_{pq}(m)$ equals the number of pseudo-lines above the crossing of $e$ and $m$.
\end{corollary}

If~$m$ does not intersect the upper envelope of $L(pq) \cup \{p\}$ at some point in negative $x$-direction of $pq$, it crosses~$q$ before crossing any of the pseudo-lines of~$L(pq)$.
Therefore, we observe:

\begin{observation}\label{obs_rank_above_envelope}
If a pseudo-line $m$ starts above every pseudo-line in~$L(pq)$, then the rank of~$m$ along $\GA{pq}$ is given by the number of pseudo-lines starting above~$m$ increased by~1, i.e., $|\{a \in \A : a \prec m \}|+1$.
\end{observation}

\subsection{Ordering Pseudo-Verticals}\label{sec_ordering_pseudo_verticals}

Given two different crossings $pq$ and $rs$ in~$\A$, it is easy to see that $\GA{pq}$ and $\GA{rs}$ may follow the same part of a pseudo-line.
Nevertheless, one can show that $\GA{pq}$ and $\GA{rs}$ will never intersect when drawn appropriately.
See \fig{fig_non_crossing_vertical} for an illustration accompanying the proof of the following lemma.

\begin{lemma}\label{lem_vertical_non_crossing}
The set of pseudo-verticals for all crossings of a pseudo-line arrangement~$\A$ can be drawn such that no two pseudo-verticals intersect.
\end{lemma}
\begin{proof}
Recall that the pseudo-verticals are fully defined by the sequence of cells they traverse.
Further, recall that each bounded cell has a unique leftmost and rightmost crossing.
For two pseudo-verticals to intersect, they have to enter a common cell~$C$.

Suppose first that $C$ is bounded.
Observe that, when traversing, say, $\GA{pq}$ in positive $x$-direction, then~$\GA{pq}$ enters~$C$ from above.
Let $C$ be between the levels $k$ and $(k+1)$.
If the current part of~$\GA{pq}$ is northbound, it enters~$C$ at the $k$-level through the pseudo-line defining the leftmost crossing of~$C$ (in~$\A$).
If it is southbound, it leaves $C$ at the $(k+1)$-level through the pseudo-line defining the rightmost crossing of~$C$.
For two pseudo-verticals to cross inside $C$, they would have to enter and leave~$C$ through four different pseudo-lines (otherwise, we could draw them without crossing in~$C$, probably changing their relative order in the next cell).
But this can only happen when one pseudo-vertical is northbound and the other is southbound in that cell (as otherwise they would either enter or leave~$C$ through the same pseudo-line), and in that case, there cannot be a crossing inside~$C$, as the pseudo-lines follow the different levels.
Once two, say, southbound rays meet in a cell (i.e., they leave a cell through the same pseudo-line of~$\A$), they follow the same pseudo-lines until reaching the south face (i.e., they pass through the same sequence of cells), and hence they can be drawn without intersecting each other.

For unbounded cells, the same argument works, with the exception that along the northbound part a pseudo-vertical enters the cell through the leftmost upper pseudo-line (i.e, at level $k$), and the southbound part leaves the cell through the rightmost pseudo-line at level~$(k+1)$.
\end{proof}

\begin{figure}
\centering
\includegraphics{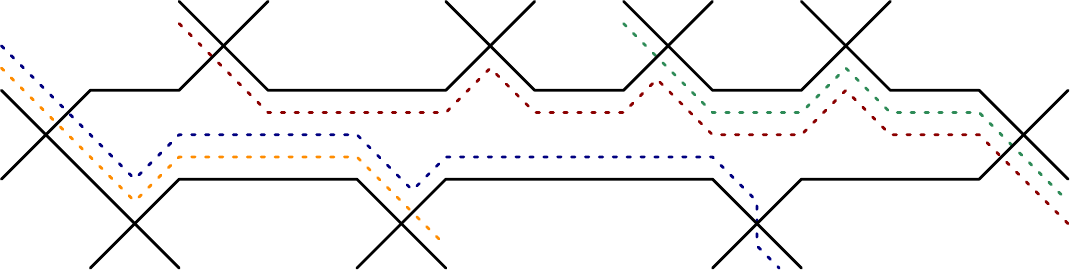}
\caption{Four pseudo-verticals (two northbound and two southbound parts) meeting in a common cell.}
\label{fig_non_crossing_vertical}
\end{figure}

\sloppypar{
An augmentation of $\A$ with a complete collection of non-intersecting pseudo-verticals defines a total order on the vertices in the arrangement (cf.~the notion of ``$P$-augmentation'' in~\cite{semispaces}).
Given an arrangement of lines, Edelsbrunner and Guibas~\cite{top_sweep,top_sweep_corrig} define a \emph{topological sweep} as a sweep of an arrangement of lines with a moving curve that intersects each line exactly once.
The topological sweep has been generalized to pseudo-line arrangements by Snoeyink and Hershberger~\cite{top_sweep_abstract}.
At any point in time during the sweep, the sweeping curve may pass over at least one crossing of the arrangement, maintaining the property that it intersects each line exactly once.
However, in contrast to a straight vertical line, there can be several crossings that may be passed next by the sweep curve.
It can be observed that we obtain the order of crossings determined by the pseudo-verticals by always sweeping over the lowest-possible crossing in a topological sweep.
In the wiring diagram shown in \fig{fig_gamma_whole_example}, the $x$-order of the crossings represents this order.
}

Lemma~\ref{lem_vertical_non_crossing} shows that we can add pseudo-verticals to an arrangement such that they only intersect at vertical infinity.
Hence, pseudo-verticals can be considered as a $P$-augmentation (cf.~\cite{semispaces}) of the initial arrangement~$\A$, and we can actually draw a wiring diagram where all the pseudo-verticals can also be represented by vertical lines.

\begin{lemma}\label{lem_order_pseudo_verticals}
The relative order of two pseudo-verticals can be obtained by a linear number of sidedness queries and queries of the form $a \prec b$.
\end{lemma}
\begin{proof}
For two pseudo-verticals~$\GA{pq}$ and~$\GA{rs}$ (we have $p \prec q$ and $r \prec s$), determining the relative order means we have to find out whether $r$ crosses $s$ before or after crossing $\GA{pq}$, i.e., whether $rs$ is to the left or to the right of~$\GA{pq}$.
The two pseudo-lines defining a crossing naturally partition the plane into four regions, which we call the upper, lower, left, and right \emph{quadrant} of the crossing.
If~$rs$ is in the left quadrant of~$pq$ (i.e., below~$p$ and above~$q$), then $rs$ is definitely to the left of~$\GA{pq}$.
Similarly, if $rs$ is in the right quadrant of $pq$ then~$rs$ is to the right of~$\GA{pq}$.
The analogous holds when exchanging the roles of $pq$ and $rs$.
Therefore, we can assume without loss of generality that $rs$ in the upper quadrant of~$pq$ (as the other case is symmetric), and that $pq$ is either in the upper or lower quadrant of~$rs$.
Consider first the case where~$pq$ is in the upper quadrant of~$rs$.
If $r \prec p$, then $r$ is part of $L(pq)$, and $rs$ is, by Lemma~\ref{lem_upper_envelope}, to the left of~$\GA{pq}$.
Analogously, if $p \prec r$, then $p$ is part of $L(rs)$, and therefore $rs$ is to the right of~$\GA{pq}$.
We are therefore left with the case where~$pq$ is in the lower quadrant of~$rs$.
If there exists a pseudo-line $a \in L(pq)$ that is above~$rs$, we again know by Lemma~\ref{lem_upper_envelope} that $rs$ is to the left of~$\GA{pq}$.
If no such pseudo-line exists, then $\rk_{pq}(r) < \rk_{pq}(s)$, and therefore, the crossing $rs$ is to the right of~$\GA{pq}$.
\end{proof}

\section{Linear-Time Pseudo-Line Selection}\label{sec_levels_algorithm}
We now discuss algorithmic properties of pseudo-verticals.
For the definition of pseudo-verticals and rank we assumed full knowledge about~$\A$.
The next task will be to make the notions accessible in the setting where we can only query the abstract order type through an oracle.
At the end we aim at using the oracle to select a pseudo-line of given rank w.r.t.\ a pseudo-vertical in $O(n)$ time.

\subsection{An Oracle for an Arrangement}\label{sec_arrangement_oracle}
Before answering queries on a pseudo-line arrangement, we need to have an internal representation of the arrangement (without explicitly building it).
Let $P$ be a predicate representing an abstract order type on a set $S$ as a counterclockwise oracle, i.e., if there is a primal point set for $S$ then $P(x,y,z)$ tells us whether the three points form a counterclockwise oriented triangle or not, in the general setting $P$ represents a chirotope.
From~\cite{extreme_journal} we borrow a linear-time procedure to determine an extreme point~$x$ of~$S$ using only queries to $P$.
We then use the following internal representation: For all $a \in S \setminus \{x\}$, we define that $x \prec a$.
For two points $a, a' \in S \setminus \{x\}$, we define that $a \prec a'$ if and only if $P(x,a,a')$, i.e., in the arrangement the crossing $ax$ precedes~$a'x$ on~$x$.
For two points $p,q \in S$ with $p \prec q$, the dual pseudo-line $r$ is below the crossing~$pq$ if and only if $P(p,q,r)$.
Hence, for three points $u,v,w \in S \setminus \{x\}$, the dual line $r$ is below the crossing defined by the (unordered) pair $(u,v)$ if and only if $P(u,v,w) = P(u,v,x)$, i.e., above/below queries for the arrangement of pseudo-lines corresponding to $P$ can be answered in constant time.
Note that a constant number of these queries also specify whether the crossing $ap$ precedes the crossing $bp$ on $p$.

Observe that, with this representation, the first unbounded cell we meet when traversing~$\GA{pq}$ against its direction is the north face, because every pseudo-line is crossed by the pseudo-line~$x$ from above (see again \fig{fig_gamma_whole_example}).
However, we will not make use of this fact in the remainder of this paper, in particular since we use the fact that the problem is symmetric when exchanging the role of the north face and the south face (which corresponds to rotating the arrangement by~$180^\circ$).

Our linear-time rank selection algorithm will depend on removing a linear fraction of the pseudo-lines in each iteration.
However, the procedure must not remove the extreme point~$x$, to keep the sub-arrangements consistent with the full arrangement.

\subsection{Selecting a Pseudo-Line}\label{sec_selecting_pseudo_line}

For a given $k$, we want to select the pseudo-line~$m$ of rank~$k$ along $\GA{pq}$.
For a subset $B$ of pseudo-lines and $m\in B$, we denote with $\rk_{pq}(m,B)$ the rank of $m$ within $B$ on $\GA{pq}$.

In the straight-line version, a linear-time selection algorithm can be used to find an element of rank~$k$ in $O(n)$ time.
This relies on the fact that the relative position of two lines can be computed in constant time.
Comparing $\rk_{pq}(s)$ and $\rk_{pq}(r)$ in the abstract setting can be reduced to deciding whether the crossing $rs$ is below some pseudo-line $a \in L(pq) \cup \{p\}$.
Doing this naively results in a linear number of queries and hence we get a selection algorithm with $\Omega(n^2)$ worst-case behavior.
We therefore need a more sophisticated method.

When discussing the relative position of two pseudo-verticals, we have seen that checking whether a crossing $rs$, $r \prec s$, is below~$\GA{pq}$ (in which case we have $\rk_{pq}(s) < \rk_{pq}(r)$), may require to determine whether $rs$ is below a pseudo-line $a \in L(pq)$, which, in the worst case, results in a linear number of comparisons.

Let $m$ be the (unknown) pseudo-line of rank $k$ within $B$.
We use a prune-and-search approach to identify~$m$.
By counting the elements of $B$ above $pq$, we determine whether $m$ is above or below $pq$ (using $O(n)$ queries).
Without loss of generality, assume $m$ is above $pq$ (the other case is symmetric) and let $U$ be the set of pseudo-lines above $pq$.
Since removing pseudo-lines from $U$ does not change the structure of the northbound part of~$\GA{pq}$, we can restrict attention to $U_B = U \cap B$.
We can also ignore (remove) pseudo-lines below $pq$ that are not in $L(pq)$, i.e., each pseudo-line~$l$ below $pq$ such that $p \prec l$.

As a next step, we can, in linear time, verify whether $m$ starts above all pseudo-lines in $L(pq) \cup \{p\}$.
If this is the case, the rank of $m$ is determined by the order in which the pseudo-lines start, and we can apply the standard selection algorithm using this order (recall Observation~\ref{obs_rank_above_envelope}).

We are therefore left with the case where $m$ starts below some element $a \in L(pq) \cup \{p\}$.
By Corollary~\ref{cor_envelope_order}, we know that we have to find the pseudo-line $e$ where $m$ crosses the upper envelope of $L(pq) \cup \{p\}$ (recall that we have $e \prec m$).

Basically, the algorithm continues as follows.
We alternatingly remove elements in $U_B$ and $L(pq)$ such that the pseudo-lines $e$ and $m$ remain in the respective set until we are left with only a constant number of pseudo-lines in the arrangement.
We describe these two pruning steps in two versions, first in a randomized version and after that in a deterministic version.
In particular the pruning of $U_B$ turns out to be much simpler in the randomized version.

\subsubsection{Randomized Pruning}
We first show how to remove pseudo-lines from~$U_B$:
Pick uniformly at random a pseudo-line $u \in U_B$.
In time proportional to the size of $L(pq)$ we find the last $a \in L(pq) \cup \{p\}$ crossed by $u$.
Since the crossing of $u$ with $\GA{pq}$ is immediately after the crossing with $a$ we can use the crossing $ua$ to split $U_B$ into elements of rank less than $\rk_{pq}(u,U_B)$ and elements of larger rank.
One of the two sets can be pruned.

Now we turn to removing pseudo-lines from~$L(pq)$:
One approach is to consider the $k$-level $\sigma$ in the sub-arrangement induced by $U_B$.
We observe that no element of $L(pq)$ can cross $\sigma$ from below before $\sigma$ crosses the upper envelope of $L(pq) \cup \{p\}$, as such a pseudo-line of $L(pq)$ would have to cross that element of $U_B$ again before $pq$.
All pseudo-lines in $L(pq)$ that start below $\sigma$ can therefore be pruned.
Among the remaining elements of $L(pq)$, the crossings with $\sigma$ define a total order.
From the remaining elements, pick, uniformly at random, a pseudo-line $b \in L(pq)$ and select (in $O(n)$ time) the pseudo-line $m' \in U_B$ where $b$ crosses $\sigma$.
We know that $m'$ is unique for the choice of~$b$.
We may prune all elements $b' \in L(pq)$ that are below $bm'$, as no element of $L(pq)$ can cross $\sigma$ more than once, and hence, no such $b'$ can be $e$ (the pseudo-line where $m$ leaves the upper envelope of $L(pq) \cup \{p\}$).
The total order on the remaining elements in $L(pq)$ implies that we can expect half of the elements to be pruned.

With the target of obtaining a deterministic version of our algorithm in mind, we describe the following alternative variant for pruning $L(pq)$.
Suppose we are given any crossing~$vw$, with $v,w \in L(pq)$ and $v \prec w$, on the upper envelope of $L(pq) \cup \{p\}$; see \fig{fig_prune_envelope}.
Let $\lv(vw)$ be the number of pseudo-lines of~$U_B$ above the crossing $vw$.
Depending on the value of $\lv(vw)$, we remove the pseudo-lines of $L(pq)$ that cannot be on the part of the upper envelope that contains the crossing with~$m$:
On $p$ consider the crossings $vp$ and $wp$.
Elements of $L(pq)$ that contribute to the upper envelope between $vw$ and $pq$ cross $p$ after $wp$.
Similarly, elements of $L(pq)$ that contribute to the upper envelope between the north face and $vw$ cross $p$ before $vp$.
Hence, depending on $\lv(vw)$, we can remove either the pseudo-lines in~$L(pq)$ that cross $p$ before~$wp$ or after~$vp$.
It remains to choose $vw$ to prune enough points.
The median~$t$ of the intersections of pseudo-lines from~$L(pq)$ with $p$ can be found with a linear number of queries (even deterministically).
Based on~$t$, we partition $L(pq)$ into the left part $L$ and the right part $R$.
Find the $\prec$-minimal element $r^\star$ of $R$ and remove all elements $l \in L$ with $r^\star \prec l$.
The removed elements do not contribute to the upper envelope of $L(pq) \cup \{p\}$.
If all elements of $L$ are removed by that, we are done.
Otherwise, we want to find the unique crossing $vw$ of a pseudo-line $v \in L$ and a pseudo-line $w \in R$ on the upper envelope of $L(pq)$.
Observe that, in the primal, $vw$ corresponds to an edge of the convex hull of $L(pq) \cup \{p\}$ that is stabbed by the supporting line of $pt$.
Finding $vw$ in linear time is described in~\cite{extreme_journal}.
For the sake of self-containment, we give a description of the randomized variant of that algorithm in terms of calls to our oracle.
We start by picking, uniformly at random, a pseudo-line $r\in R$ and then determine the last crossing $rl$ with a pseudo-line from $L$ on $r$.
It can be argued (cf.~Lemma~\ref{lem:prune_LR}) that every pseudo-line $r' \in R$ whose crossing with $l$ is behind the crossing $rl$ fails to be a candidate for $w$ and can hence be removed.
We expect to remove half of the pseudo-lines from $R$ through this.
A symmetric step can be used to reduce the size of $L$.
By always applying the reduction to the larger of the two we obtain a procedure that outputs the pair $vw$ with expected $O(|L(pq)|)$ queries.
Next we determine which elements of $U_B$ are above and which below $vw$.
From this we deduce whether the element $m\in U_B$ with $\rk_{pq}(m,U_B)=k$ intersects $\GA{pq}$ in the part where it follows the envelope of $L$ or where it follows the envelope of $R$.
Depending on this we can either prune $L$ or $R$ from $L(pq)$.
It remains to show that the expected number of queries is in~$O(n)$.
\begin{figure}[ht]
\centering
\includegraphics{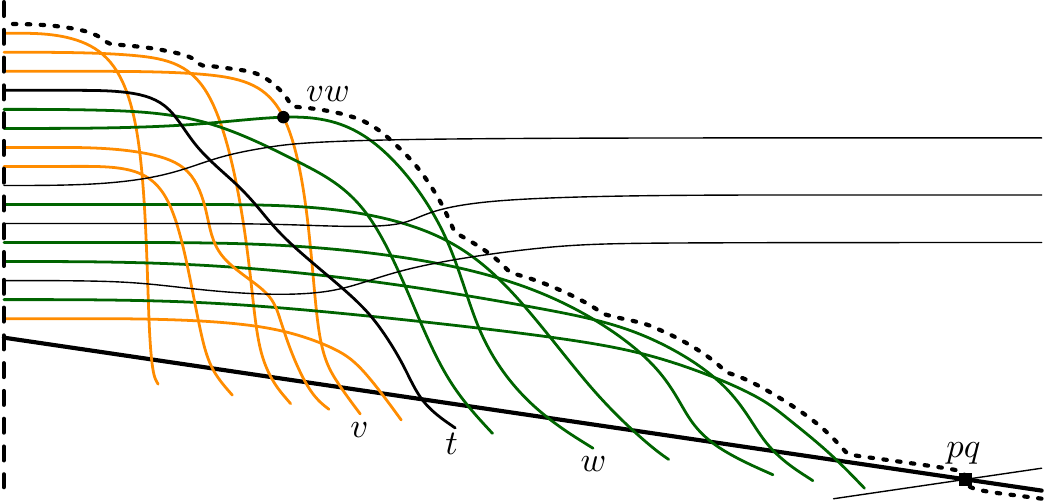}
\caption{Partitioning the pseudo-lines in $L(pq)$ along~$p$ by a pseudo-line~$t$.}
\label{fig_prune_envelope}
\end{figure}

We begin with the analysis of the expected number of oracle calls.
The justification of some of the claims that are made is in the lemmas below.
(They are closely related to arguments used in~\cite{extreme_journal} in the primal, and are stated here in terms of queries to our oracle for the sake of self-containment.)

\begin{theorem}
Given an arrangement~$\A$ of pseudo-lines, a subset~$B$ of its pseudo-lines, a crossing $pq$, and a natural number $k \leq |B|$, the pseudo-line~$m \in B$ with $\rk_{pq}(m,B) = k$, can be found with a randomized algorithm that uses an expected linear number of calls to the oracle representing $\A$.
\end{theorem}
\begin{proof}
Deciding whether $m$ is above or below $pq$ and computing the initial sets $U_B$ and $L(pq)$ can be done deterministically with a linear number of queries.

A pruning step for $U_B$ requires $O(|U_B|+|L(pq)|)$ queries.
We prune $U_B$ only if $|L(pq)| \leq |U_B|$ and in expectation we prune at least one quarter of the elements of $U_B$.
Therefore, the expected number of queries to select $m$ from $B$ is linear in $|B|$.

A pruning step for $L(pq)$ starts with the median computation on $p$.
This requires $O(|L(pq)|)$ queries in expectation.
After that we have the sets $L$ and~$R$ and again with a linear number of queries we make sure that $l \prec r$ for all $l \in L$ and $r \in R$.
Iteratively prune the larger of $L$ and $R$ with $O(|L|+|R|)$ queries.
Upon pruning the expected size of the set halves (Lemma~\ref{lem:exp-pru-size}).
Hence, for pruning steps on $L$ we expect to use $O(|L|)$ queries and symmetrically for $R$.
Finally, we need $O(|U_B|)$ to decide which of $L$ and $R$ can be discarded.
Since $|U_B|\leq |L(pq)|$ we conclude that halving the size of $L(pq)$ can be done with $O(|L(pq)|)$ queries.
Thence the total number of queries used in the pruning of $L(pq)$ is expected to be in $O(|L(pq)|)$.
\end{proof}

\begin{lemma}\label{lem:prune_LR}
Let $l$ be the last pseudo-line crossing $r$ and let $r'$ be a pseudo-line crossing $l$ after the crossing $rl$, then $r' \neq w$.
\end{lemma}
\begin{proof}
Suppose there is a pseudo-line $a \in L$ such that $ar'$ is the crossing on the upper envelope.
Then it follows that the crossing $l r$ is above $r'$ and because $a \prec r$ the crossing $ar$ is below $r'$.
This, however implies that on $r$ the crossing with $l$ precedes the crossing with $a$.
This is a contradiction to the choice of $l$.
\end{proof}

The dual pruning of $L$ is as follows:
Pick uniformly at random a pseudo-line $l\in L$ and determine the first crossing $rl$ with a pseudo-line from $R$ on $l$.
From Lemma~\ref{lem:prune_LR-2} we obtain that every pseudo-line $l'\in L$ whose crossing with $r$ precedes the crossing $rl$ fails to be a candidate for $v$ and can hence be removed.

\begin{lemma}\label{lem:prune_LR-2}
Let $r$ be the first pseudo-line crossing $l$ and let $l'$ be a pseudo-line crossing $r$ before the crossing $l r$, then $l' \neq v$.
\end{lemma}
\begin{proof}
Suppose there is a pseudo-line $b \in R$ such that $l'b$ is the crossing on the upper envelope.
Then it follows that the crossing $l'b$ is above $l$ and because $l \prec r$ the crossing $l'r$ is below $l$.
This, however implies that on $l$ the crossing with $b$ precedes the crossing with $r$.
This is a contradiction to the choice of $r$.
\end{proof}

\begin{lemma}\label{lem:complete}
For every pair $(l,l')$ of pseudo-lines in $L$ we have: if $l'$ is not removed when choosing $l$ as the (random) pseudo-line for pruning, then $l$ is removed when $l'$ is chosen as the (random) pseudo-line for pruning.
\end{lemma}
\begin{proof}
If the same pseudo-line $r \in R$ is the first to cross $l$ and $l'$, then the statement is obvious.
If $r$ is the first to cross $l$ and $r'$ is the first on $l'$, then there must be a crossing of $rr'$ between $l r$ and $l r'$ on $r$.
Hence on $r'$ the crossing with $l$ precedes the crossing with $l'$.
\end{proof}

\begin{lemma}\label{lem:exp-pru-size}
The expected size of the set obtained with a pruning step from $L$ is at most $|L|/2$.
\end{lemma}
\begin{proof}
On the set $L$ we define a directed graph $G_L$ with edges $l \to l'$ if $l'$ is removed when choosing $l$ as the (random) pseudo-line for pruning.
Lemma~\ref{lem:complete} implies that $G_L$ is a tournament.
The expected number of pseudo-lines removed in the pruning equals the expected out-degree plus one.
The precise value is $(|L|+1)/2$ so in expectation less than $|L|/2$ pseudo-lines remain.
\end{proof}

The analogous statements of Lemma~\ref{lem:complete} and Lemma~\ref{lem:exp-pru-size} for the set~$R$ are true as well.

\paragraph*{Remark.} Note that, by removing pseudo-lines from $L(pq)$, we obtain a new arrangement~$\A'$, in which the pseudo-vertical~$\GA{pq}$ will, in general, follow different pseudo-lines from $L(pq)$ along its northbound ray.
Still, the number of pseudo-lines above the crossing $em$ that we look for remains the same, and $m$ will have the same rank with respect to the new pseudo-vertical $\GA{pq}'$ in~$\A'$.

\subsubsection{Deterministic Pruning}
To remove pseudo-lines from~$L(pq)$ deterministically, we can directly apply the deterministic algorithm given in~\cite{extreme_journal} to find~$vw$ after~$O(n)$ queries.

Recall that to remove elements of~$U_B$, we pick a pseudo-line~$u \in U_B$.
We compute the rank $\rk_{pq}(u,B)$ by finding the corresponding pseudo-line~$b \in L(pq) \cup \{p\}$ at which $u$ passes through the upper envelope of $L(pq) \cup \{p\}$.
Clearly, this can be done in linear time using our basic operations.
If $u = m$, we are done.
If $\rk_{pq}(u,B) < k$, then all pseudo-lines in $U_B$ below $bu$ can be removed (we will see how to choose~$u$ to remove a constant fraction of the pseudo-lines).
Otherwise, we remove all pseudo-lines in $U_B$ above $bu$ and update~$k$ accordingly.
While by this operation we obtain a new arrangement~$\A'$, the northbound ray of~$\GA{pq}$ in~$\A'$ is defined by the same pseudo-lines as in~$\A$, and we can therefore safely continue with the next iteration.
It remains to show how to pick $u$ in a deterministic way such that at least a constant fraction of~$U_B$ can be removed.
To this end, we use the concept of $\varepsilon$-approximation of range spaces.

Our definitions follow~\cite{new_trends_matousek}.
A \emph{range space} is a pair $\Sigma = (X, \mathcal{R})$ where $X$ is a set and $\mathcal{R}$ is a set of subsets of~$X$.
The elements of $\mathcal{R}$ are called \emph{ranges}.
For $X$ being finite, a subset $A \subseteq X$ is an \emph{$\varepsilon$-approximation} for $\Sigma$ if, for every range $R \in \mathcal{R}$, we have
\[
\abs{\frac{\abs{A \cap R}}{\abs{A}} - \frac{\abs{X \cap R}}{\abs{X}}} \leq \varepsilon \enspace .
\]

A subset $Y$ of $X$ is \emph{shattered} by $\mathcal{R}$ if every possible subset of~$Y$ is a range of~$Y$.
The \emph{Vapnik-Chervonenkis dimension} (VC-dimension) of $\Sigma$ is the maximum size of a shattered subset of $X$.
For sets with finite VC-dimension, Vapnik and Chervonenkis~\cite{vapnik_chervonenkis} give the following seminal result.

\begin{theorem}[Vapnik, Chervonenkis~\cite{vapnik_chervonenkis}]\label{finite_vc}
Any range space of VC-di\-men\-sion $d$ admits an $\varepsilon$-approximation of size $O(d/\varepsilon^2 \log(d/\varepsilon))$.
\end{theorem}

For $|X| = n$, the \emph{shatter function} $\pi_\mathcal{R}(n)$ of a range space $(X, \mathcal{R})$ is defined by
\[
\pi_\mathcal{R}(n) = \max\{\abs{\{ Y \cap R : R \in \mathcal{R} \}} : Y \subseteq X\} \enspace .
\]
Vapnik and Chervonenkis~\cite{vapnik_chervonenkis} show that, for a range space $(X, \mathcal{R})$ of VC-dimension $d$, $\pi_\mathcal{R}(n) \in O(n^d)$ holds.
Matou\v{s}ek~\cite{matousek_approximations_conf,matousek_approximations_journal} gives, for a constant~$\varepsilon$, a linear-time algorithm for computing a constant-size $\varepsilon$-approximation for range spaces of finite VC-dimension $d$ (simplified by Chazelle and Matou\v{s}ek~\cite{chazelle_matousek}), provided there exists a subspace oracle.

\begin{definition}\label{def_subspace_oracle}
A \emph{subspace oracle} for a range space $(X, \mathcal{R})$ is an algorithm that returns, for a given subset $Y \subseteq X$, the set of all distinct intersections of $Y$ with the ranges in $\mathcal{R}$, i.e., the set $\{Y \cap R : R \in \mathcal{R}\}$ and runs in time $O(|Y|\cdot h)$, where $h$ is the number of sets returned.
\end{definition}

\begin{theorem}[{Matou\v{s}ek~\cite[Theorem~4.1]{matousek_approximations_conf}}]\label{thm_vc_linear_time}
Let $\Sigma = (X,\mathcal{R})$ be a range space with the shatter function $\pi_\mathcal{R}(n) \in O(n^d)$, for a constant $d \geq 1$.
Given a subspace oracle for $\Sigma$ and a parameter $r \geq 2$, a $(1/r)$-approximation for $\Sigma$ of size $O(r^2 \log r)$ can be computed in time $O(|X|(r^2 \log r)^d)$.
\end{theorem}

Observe that, for such a range space, the running time of the subspace oracle is bounded by $O(|Y|^{d+1})$, as $h$ is at most $\pi_\mathcal{R}(|Y|)$.

Suppose that, e.g., $X$ is a point set in the Euclidean plane and $\mathcal{R}$ consists of all possible subsets of $X$ defined by half-planes, defining a range space $\Sigma = (X, \mathcal{R})$.
Hence, $\mathcal{R}$ is the set of \emph{semispaces} defined by the order type of~$X$.
The VC-dimension of $(X,\mathcal{R})$ is known to be~$3$~\cite{new_trends_matousek}. %
Hence, a constant-size $\varepsilon$-approximation of a point set for~$\mathcal{R}$ exists.
(As pointed out in~\cite{lo_matousek_steiger}, this approximation allows constructing an approximate ham-sandwich cut in constant time, such that on every side of the cut there are no more than $1/2 + \varepsilon$ of the points of each class.)
The subspace oracle returns, for any subset~$Y$ of points, all possible ways a line can separate $Y$, which can easily be done in time $O(\abs{Y}^3)$.

The VC-dimension of~3 for that range space holds also for abstract order types (see also~\cite{gaertner_welzl}).
A subspace oracle for semispaces of a given set can easily be implemented using the definition of a semispace by allowable sequences~\cite{semispaces};
for each pair $f,g \in Y, f \prec g,$ in the dual pseudo-line arrangement, we report the pseudo-lines above the crossing $fg$ and, say, $f$ as a semispace, and the pseudo-lines below $fg$ and $g$ as a second semispace.

Consider again the set $U_B$ of pseudo-lines above the crossing $pq$ in~$\A$.
Using Theorem~\ref{thm_vc_linear_time}, we obtain an $\varepsilon$-approximation $A \subset U$ for the range space of semispaces, i.e., the pseudo-lines of~$U$ above and below a point in~$\A$.
$A$ is of constant size for a fixed~$\varepsilon$.
For each pseudo-line $o \in A$, we obtain the crossing of~$o$ with the pseudo-lines in~$L(pq)$ that defines the rank $\rk_{pq}(o)$ in $O(n)$ time.
Let $u_A \in A$ have the median rank among the elements of~$A$.
Then not less than $1/2-\varepsilon$ pseudo-lines of $U$ are above and below the crossing $bu_A$ on the upper envelope of $L(pq)$;
we may prune a constant fraction of the elements in~$U_B$.

\subsubsection{Analysis}
In each iteration, our problem consists of the remaining pseudo-lines in~$U$ and in~$L(pq)$, plus a constant number of additional pseudo-lines (i.e., $p, q,$ and the pseudo-lines needed by the oracle, in our case the pseudo-line~$x$).
Let~$n$ be the number of these pseudo-lines.
In each iteration we prune the larger of $U$ and $L(pq)$.
In both cases, we remove at least half of the pseudo-lines on one side of $pq$, and therefore $n/4-O(1)$ pseudo-lines in each iteration.
Since each iteration takes $O(n)$ time, we have overall a linear-time prune-and-search algorithm.

\begin{theorem}\label{thm_slope_abstract}
Given an arrangement~$\A$ of pseudo-lines, a subset~$B$ of its pseudo-lines, a crossing $pq$, and a natural number $k \leq |B|$, the pseudo-line~$m \in B$ of rank~$k$ in $B$ on the pseudo-vertical through $pq$, i.e., $\rk_{pq}(m,B)$, can be found in linear time using only sidedness queries on the corresponding abstract order type.
\end{theorem}

\paragraph*{Remark.}
Several years before the linear-time algorithm for ham-sandwich cuts~\cite{lo_matousek_steiger} was developed, Megiddo~\cite{megiddo} considered the following restricted version of the ham-sandwich cut problem.
Given a set of red and a set of blue points with disjoint convex hulls, find a line that bisects both the red and the blue point set.
Actually, the resulting line does not have to be a bisector, but the number of red and blue points on one side of the line can be chosen arbitrarily.
In the dual representation, we are given $m$ blue lines with positive slope and $n$ red lines with negative slope, and we want to find the intersection point between a $k_1$-level in the blue lines and the $k_2$-level in the red lines.
If we consider the pseudo-lines in $L(pq) \cup \{p\}$ as red pseudo-lines and the pseudo-lines in $U$ as the blue ones, we are looking for the intersection point between the $k$-level in $U$ and the 1-level in $L(pq) \cup \{p\}$.
However, Megiddo's algorithm also depends on the realization of the line arrangement;
the algorithm requires selecting the median of a subset of crossings ordered by their $x$-coordinate and selecting the intersections of a given rank at vertical lines.
These problems also have to be solved when abstracting the general ham-sandwich cut algorithm.

\section{Revisiting the Ham-Sandwich Cut Algorithm}
\label{sec_ham_sandwich_revisited}

In this section, we describe an application of pseudo-verticals for a bisection algorithm, namely the linear-time ham-sandwich cut algorithm by Lo, Matou\v{s}ek, and Steiger~\cite{lo_matousek_steiger} (the LMS algorithm).
To this end, we revisit the description given in~\cite{lo_matousek_steiger}; we adapt some of the terminology (like replacing ``line'' with ``pseudo-line'') and argue for the correspondence between entities in the original description and their abstract counterpart.

Lo, Matou\v{s}ek, and Steiger~\cite{lo_matousek_steiger} describe two different variants of the algorithm, one for points in the plane and the other one for points in arbitrary dimension (their work is a generalization of a 2-dimensional version by Lo and Steiger presented in~\cite{lo_steiger}).
For the 2-dimensional case, a result by Matou\v{s}ek~\cite[Theorem~3.2]{matousek_construction} is used for appropriately selecting a set of vertical lines.
In higher dimensions, they use a different approach (given in~\cite{matousek_approximations_conf}) based on an $\varepsilon$-approximation with the ranges being defined as sets of hyperplanes that are stabbed by segments (we will give a formal definition later).
While the higher-dimensional variant appears to be less instructive, it is easier to apply to our setting.
We therefore will use this variant for our 2-dimensional setting; here, we do not give the description for arbitrary dimension, but transcribe it to dimension~2 only.
Still, our exposition closely follows~\cite{lo_matousek_steiger}, while merely pointing out the parts where the applicability to our abstract setting might not be obvious.

Let $P$ be a finite set of $n$ points in the Euclidean plane.
A line~$h$ \emph{bisects}~$P$ if no more than $n/2$ points lie in either of the open half-planes defined by~$h$.
We call $h$ a \emph{bisector}.
If $P$ is a disjoint union of two point sets $P_1, P_2$, a \emph{ham-sandwich cut} is a line that simultaneously bisects both $P_1$ and $P_2$ (a \emph{red} and a \emph{blue} set).
This definition extends to abstract order types in a natural way.
It is well-known that a ham-sandwich cut always exists.
Let $T$ be an interval on the $x$-axis, and let $V(T)$ be the vertical slab between the two vertical lines defining~$T$.
The interval has the \emph{odd intersection property} with respect to the levels $\lambda_1$ and $\lambda_2$ if $|(\lambda_1 \cap \lambda_2) \cap V(T)|$ is odd.
If $k = \floor{(n+1)/2}$, the $k$-level is called \emph{median level}.
In our case, each slab is defined by two pseudo-verticals.

The algorithm works in a prune-and-search manner.
Let us first consider the setting where we are given an actual set of points in~$\euclid^2$.
In every iteration, we are given
\begin{itemize}
 \item an interval $T$ on the $x$-axis,
 \item two sets $G_1$ and $G_2$ of lines dual to a subset of points in $P_1$ and $P_2$, respectively, with $|G_1| = n_1$ and $|G_2| = n_2$, and
 \item two integers $k_1$ and $k_2$, with $1 \leq k_1 \leq n_1$ and $1 \leq k_2 \leq n_2$, denoting the $k_1$-level $\lambda_1$ and the $k_2$-level~$\lambda_2$, respectively.
\end{itemize}
Further, we know that $T$ has the odd intersection property for the $k_1$-level and the $k_2$-level.
We denote the arrangements corresponding to $G_1$ and $G_2$ with $\A_1$ and $\A_2$, respectively.
Initially, $\lambda_1$ and $\lambda_2$ are the median levels of the two arrangements (for which the odd intersection property holds).
Without loss of generality, suppose $n_1 \geq n_2$.
The algorithm consists of the following four steps:
\begin{enumerate}
 \item Divide $T$ into a constant number of subintervals $T_1, \dots, T_C$, to limit the number of pseudo-lines that are on $\lambda_1$ within each subinterval.
 \item Find a subinterval $T_j$ with the odd intersection property.
 \item Construct a trapezoid $\tau \subset V(T_j)$ such that
 \begin{enumerate}
  \item $\lambda_1 \cap V(T_j) \subset \tau$, and
  \item at most half of the lines of $\A_1$ intersect $\tau$.
 \end{enumerate}
 \item Discard the lines of $\A_1$ that do not intersect $\tau$, update $k_1$ accordingly and continue within the interval~$T_j$.
\end{enumerate}

In our abstract setting, the interval~$T$ is given by a pair of pseudo-verticals.
Recall that there is a total order on the pseudo-verticals of a pseudo-line arrangement.
The trapezoid~$\tau$ will also be replaced by a corresponding structure that will be described later.

\subsection{Obtaining Intervals}
Step~1 in the algorithm is the one that is technically most involved.
The straight-line version can be solved using the following result by Matou\v{s}ek.%
\footnote{Lo, Matou\v{s}ek, and Steiger~\cite{lo_matousek_steiger} refer to~\cite{matousek_construction}, where~\cite[Lemma~4.5]{haussler_welzl} (Lemma~\ref{lem_corridors} herein) is used, and also refer to~\cite{matousek_approximations_conf} in this context, where a general algorithm for constructing $\varepsilon$-approximations is given.
}

\begin{lemma}[Matou\v{s}ek]\label{lem_segment_ranges}
Let $H$ be a collection of $n$ hyperplanes in $\euclid^d$ and let $\mathcal{R}$ be all subsets of $H$ of the form $\{h \in H : h \cap s \neq \emptyset \}$ where $s$ is a segment in~$\euclid^d$.
An $\varepsilon$-approximation for the range space $(H,\mathcal{R})$ of size $O(\varepsilon^{-2} \log (1/\varepsilon))$ can be computed in time $O(f(\varepsilon)n)$, where $f(\varepsilon)$ is a factor depending on $\varepsilon$ and $d$ only.
\end{lemma}

Let us go into the details why this lemma also holds for arrangements of pseudo-lines.
To this end, a general result by Haussler and Welzl~\cite{haussler_welzl} is used.
\begin{lemma}[{Haussler, Welzl~\cite[Lemma~4.5]{haussler_welzl}}]\label{lem_corridors}
Assume $k \geq 1$ and $(X,\mathcal{R})$ is a range space of VC-dimension $d \geq 2$.
Let~$\mathcal{R}'$ be the set of all sets of the form $\bigcup_{i=1}^k R_i - \bigcap_{i=1}^k R_i$, where $R_i$ is a range in~$\mathcal{R}$, $1 \leq i \leq k$.
Then $(X, \mathcal{R}')$ has VC-dimension less than $2dk \log(dk)$.
\end{lemma}

We already discussed that a range space defined by the semispaces of an abstract order type has VC-dimension~3.
We can combine this fact with Lemma~\ref{lem_corridors} in the following way.
Consider two semispaces $S_1$ and $S_2$ of an abstract order type, defined by the pseudo-lines above two points $p_1$ and $p_2$ in the corresponding pseudo-line arrangement~$\A$ (for simplicity, suppose that none of $p_1$ and $p_2$ lie on a pseudo-line of~$\A$).
Let $R = (S_1 \cup S_2) \setminus (S_1 \cap S_2)$.
Then $R$ consists exactly of the pseudo-lines that separate $p_1$ from $p_2$.
By the Levi Enlargement Lemma (see \cite{levi}), we can extend~$\A$ by a pseudo-line through any two points, and therefore obtain a pseudo-line $\chi$ for~$\A$ containing both $p_1$ and $p_2$.
Consider the part of $\chi$ between $p_1$ and $p_2$.
We call such a part a \emph{pseudo-segment}.
The pseudo-lines crossed by this pseudo-segment are exactly those in~$R$.
Applying Lemma~\ref{lem_corridors}, we can therefore obtain a range space that is defined by the pseudo-lines that can be crossed by pseudo-segments from the range space defined by the semispaces;
this new range space has again finite VC-dimension.
(Note that, while we explained the application of Lemma~\ref{lem_corridors} using points $p_1$ and $p_2$, the argument also holds for pseudo-segments defined by crossings of~$\A$, as the endpoints of the pseudo-segment can be perturbed to be in one of the four cells adjacent to a crossing.)
Using Matou\v{s}ek's linear-time algorithm for obtaining an $\varepsilon$-approximation~\cite[Theorem~4.1]{matousek_approximations_conf}, we can state the following counterpart to Lemma~\ref{lem_segment_ranges} for abstract order types in the plane.

\begin{corollary}\label{cor_corridors_abs}
Let $\Sigma = (X, \mathcal{R})$ be a range space where $X$ is the set of pseudo-lines in a pseudo-line arrangement~$\A$ and $\mathcal{R}$ consists of the sets of pseudo-lines of~$\A$ that are crossed by pseudo-segments obtained on~$\A$.
Then a $(1/r)$-approximation of constant size for $\Sigma$ can be computed in $O(|X|)$ time for a given $r \geq 2$.
\end{corollary}

For future reference, let us call this range space the \emph{pseudo-segment range space} of the arrangement.

For the ham-sandwich cut algorithm, we can now proceed in the following way.
Using Corollary~\ref{cor_corridors_abs}, we obtain an $\varepsilon$-approximation $A$ for the pseudo-segment range space of~$\A_1$;
we choose $\varepsilon = 1/12$ with foresight.
Sort the crossings of $A$ by the order implied by the pseudo-verticals through the crossings on the original arrangement~$\A$.
Since $A$ is of constant size, this can be done in~$O(|\A|)$ time.
We use Theorem~\ref{thm_slope_abstract} to determine the $k_1$-level of $\A_1$ at the pseudo-vertical of each crossing in~$A$.
Hence, for each pseudo-vertical, we get a crossing in $\A$ that has level $k_1$ in $\A_1$.
Counting the elements of $\A_2$ above each such crossing allows us to find a crossing $pq$ and a crossing $p'q'$ consecutive in~$A$ with the odd intersection property.
We again use Theorem~\ref{thm_slope_abstract} to select the pseudo-lines that are of rank $k_1 - c\varepsilon n_1$ and $k_1 + c \varepsilon n_1$ in~$\A_1$ at the pseudo-vertical $\GA{pq}$, and do the same at $\GA{p'q'}$ for a constant $c$ (we fix $c = 3/2$ with foresight).
Hence, we have six crossings in $\A$ of which we know the level within $\A_1$.
Let $g_\mathrm{l}$ and $g_\mathrm{r}$ be the crossings at the $k_1$-level along $\GA{pq}$ and $\GA{p'q'}$, respectively.
We denote the crossings at the $(k_1 - c \varepsilon n_1)$-level by $d_\mathrm{l}^-$ and $d_\mathrm{r}^-$.
Their counterparts at the $(k_1 + c \varepsilon n_1)$-level are denoted by $d_\mathrm{l}^+$ and $d_\mathrm{r}^+$.

\subsection{Properties of a Trapezoid-Like Structure}
In the original LMS algorithm~\cite{lo_matousek_steiger}, the points at the given levels were determined by the intersections of the levels with the vertical lines.
These points formed a trapezoid.
However, the actual properties used are the ones of the points and not the ones of the trapezoid as a geometric object.
In this part, we reproduce the line of arguments used in~\cite{lo_matousek_steiger} to show that at least half of the pseudo-lines in~$\A_1$ are either above both $d_\mathrm{l}^-$ and $d_\mathrm{r}^-$ or below both $d_\mathrm{l}^+$ and $d_\mathrm{r}^+$, and that these pseudo-lines are not on the $k_1$-level between $g_\mathrm{l}$ and $g_\mathrm{r}$.

Consider the arrangement $\A_1$.
We bound the number of pseudo-lines that separate $d_\mathrm{l}^-$ from $d_\mathrm{r}^-$, i.e., the pseudo-lines crossing a pseudo-segment between $d_\mathrm{l}^-$ and $d_\mathrm{r}^-$.
The levels of $d_\mathrm{l}^-$ and $d_\mathrm{r}^-$ are the same.
Therefore, the numbers of pseudo-lines of the approximation~$A$ above these two crossings differ by at most $2\varepsilon\abs{A}$.
If there would be more than $2\varepsilon\abs{A}$ pseudo-lines of $A$ separating $d_\mathrm{l}^-$ from $d_\mathrm{r}^-$, then at least one of these pseudo-lines would have to be above $d_\mathrm{l}^-$ and below $d_\mathrm{r}^-$, and another one would have to be below $d_\mathrm{l}^-$ and above $d_\mathrm{r}^-$.
Hence, the crossing between these two pseudo-lines would have to be in the interval between $\GA{pq}$ and $\GA{p'q'}$.
But this contradicts the choice of $\GA{pq}$ and $\GA{p'q'}$, as there is no pseudo-vertical through a crossing of two pseudo-lines of~$A$ between them.
Hence, any pseudo-segment between $d_\mathrm{l}^-$ and $d_\mathrm{r}^-$ crosses at most $2\varepsilon\abs{A}$ of the pseudo-lines in~$A$.
By the $\varepsilon$-approximation property, at most $3\varepsilon n_1$ pseudo-lines of $\A_1$ intersect such a pseudo-segment.

Suppose there is a pseudo-line~$w$ of $\A_1$ that is above both $d_\mathrm{l}^-$ and $d_\mathrm{r}^-$, but still $w$ is an element of the $k_1$-level between $g_\mathrm{l}$ and $g_\mathrm{r}$.
The part of the arrangement where this can happen is bounded by $\GA{pq}$ and $\GA{p'q'}$.
Then also any pseudo-segment~$s$ between $d_\mathrm{l}^-$ and $d_\mathrm{r}^-$ would have to cross the relevant part of the $k_1$-level (recall that the pseudo-segment can be considered as a part of a pseudo-line in an extended arrangement).
At both $d_\mathrm{l}^-$ and $d_\mathrm{r}^-$, the pseudo-segment~$s$ is at level $(k_1 - c \varepsilon n_1)$, and therefore has to cross $2(k_1 - (k_1 - c \varepsilon n_1)) = 2c\varepsilon n_1$ pseudo-lines to reach the $k_1$-level and then return to level $(k_1 -c \varepsilon n_1)$.
By the choice of $c$, this is exactly $3\varepsilon n_1$, the maximum number of crossings the pseudo-segment~$s$ can have.
Hence, $s$ cannot go below the $k_1$-level and therefore the pseudo-line~$w$ cannot intersect the $k_1$-level.
For our prune-and-search approach, we can therefore remove all pseudo-lines of $\A_1$ that are above both $d_\mathrm{l}^-$ and $d_\mathrm{r}^-$, and, by symmetric arguments, can do the same for the ones below both $d_\mathrm{l}^+$ and~$d_\mathrm{r}^+$.

It remains to count how many pseudo-lines are removed.
There are exactly $2c \varepsilon n_1$ pseudo-lines separating $d_\mathrm{l}^+$ from $d_\mathrm{l}^-$, as well as $d_\mathrm{r}^+$ from $d_\mathrm{r}^-$.
Further, we argued that there are at most $3\varepsilon n_1$ pseudo-lines between $d_\mathrm{l}^-$ and $d_\mathrm{r}^-$, as well as between $d_\mathrm{l}^+$ and $d_\mathrm{r}^+$.
This amounts to $(4c+6) \varepsilon n_1 = 12 \varepsilon n_1$, where each pseudo-line is counted twice.
Therefore, we have to keep at most $6 \varepsilon n_1 = n_1/2$ pseudo-lines of $\A_1$.

\subsection{A Note on the Intervals}
After having pruned a linear fraction of pseudo-lines, the algorithm performs another iteration within a smaller interval for which the odd-intersection property holds.
Note that we need to continue within this interval, as the $k_1$-level (for an updated $k_1$) equals the median-level only within that region.
In the geometric variant, the interval was explicitly given by the vertical lines of the current slab.
For pseudo-verticals, we have no such fixed position, and actually the pseudo-verticals will, in general, be different when pseudo-lines are removed from the arrangement (even the relative order of two crossings may change).
However, we can safely define the interval for the subproblem by the two crossings $g_\mathrm{l}$ and~$g_\mathrm{r}$, as the odd intersection property can be seen as a property of two points on one of the two levels (only the number of pseudo-lines of $\A_2$ above each of the two points is relevant here).
It is interesting to observe that also $\GA{g_\mathrm{l}}$ and $\GA{g_\mathrm{r}}$ do not have a different relative position in the new arrangement.

\paragraph*{}
With the proof of Theorem~\ref{thm_slope_abstract} and the observations implying that all the remaining parts of the LMS algorithm actually do not depend on the geometric realization, we obtain Theorem~\ref{thm_abstract_ham_sandwich} as our main result.

\section{Conclusion}
In this paper, we defined a possible replacement of a vertical line in line arrangements for arrangements of pseudo-lines and showed that it fulfills important algorithmic properties.
In particular, we were able to show how to select the $k$th pseudo-line crossed by this pseudo-vertical line and the crossing where this pseudo-line (locally) enters the $k$-level in linear time, using only sidedness queries.
As an application, we showed how these pseudo-vertical lines replace vertical lines in the linear-time ham-sandwich cut algorithm by Lo, Matou\v{s}ek, and Steiger~\cite{lo_matousek_steiger}.

In essence, the order of the pseudo-verticals through all crossings of a pseudo-line arrangement fix one specific allowable sequence for an abstract order type.
Theorem~\ref{thm_slope_abstract} allows us to select certain elements of a permutation in that allowable sequence in linear time.
We have seen that this approach is a generalization of the result presented in~\cite{extreme_journal}.
(Note that the oracle also uses the extreme point~$x$, which, in turn, requires applying Theorem~\ref{thm_slope_abstract}; any internal representation that can answer queries of the form $a \prec b$ in constant time allows us to find an extreme point in linear time.)

Compared to $L(pq)$, pruning $U_B$ deterministically required the technically involved step of selecting an $\varepsilon$-approximation.
Is there a more ``light-weight'' deterministic way to prune $U_B$, similar to the one used for $L(pq)$?

\bibliographystyle{abbrv}
\bibliography{bibliography}
\end{document}